\documentclass[fleqn,11pt]{article}
\usepackage{fullpage}
\usepackage{a4wide}
\usepackage{amsmath,amssymb,amsfonts,amsthm}
\usepackage[fleqn,tbtags]{mathtools}
\usepackage[english]{babel}

\usepackage{graphicx}
\usepackage{epsfig}
\usepackage{color}
\usepackage{xspace}
\usepackage{url}
\usepackage{ifpdf}
\ifpdf
\usepackage{hyperref}
\usepackage[utf8]{inputenc} % allow utf-8 input
\usepackage{booktabs}       % professional-quality tables
\else
\usepackage[hypertex]{hyperref}
\fi
\usepackage [autostyle, english = american]{csquotes}
\usepackage[T1]{fontenc}
\MakeOuterQuote{"}

\newtheorem{theorem}{Theorem}%[section]
\newtheorem{lemma}[theorem]{Lemma}

\newtheorem{definition}{Definition}%[section]

\newtheorem{claim}[theorem]{Claim}

\newtheorem{conjecture}{Conjecture}
\makeatletter
\newtheorem*{rep@theorem}{\rep@title}
\newcommand{\newreptheorem}[2]{%
\newenvironment{rep#1}[1]{%
 \def\rep@title{#2 \ref{##1}}%
 \begin{rep@theorem}}%
 {\end{rep@theorem}}}
\makeatother

\newreptheorem{theorem}{Theorem}
\newreptheorem{lemma}{Lemma}

\newcommand{\R}{\mathbb{R}}
\newcommand{\In}{\textrm{In}}

\newcommand{\eps}{\varepsilon}
\newcommand{\from}{\colon}
\newcommand {\ignore} [1] {}

\newcommand{\etal}{{\em et al.\ }\xspace}

\newcommand*\samethanks[1][\value{footnote}]{\footnotemark[#1]}

\title{Lower Bounds for Multiplication via Network Coding}

\author{Peyman Afshani\thanks{Aarhus University.
Email: \texttt{peyman@cs.au.dk}.} \and Casper Freksen\thanks{Aarhus University. Supported by a Villum Young Investigator Grant. Email: \texttt{\{cfreksen, lior.kamma\}@cs.au.dk}.} \and Lior Kamma\samethanks \and Kasper
  Green Larsen\thanks{Aarhus University. Supported by a Villum Young
    Investigator Grant and an AUFF Starting Grant.
Email: \texttt{larsen@cs.au.dk}.}}

\begin{document}

\date{}
\maketitle

\begin{abstract}
Multiplication is one of the most fundamental computational problems,
yet its true complexity remains elusive. The best known upper bound,
by F\"{u}rer,
shows that two $n$-bit numbers can be multiplied via a boolean circuit of size
$O(n \lg n \cdot 4^{\lg^*n})$, where $\lg^*n$ is the very slowly growing
iterated logarithm. In this work, we prove that if a central
conjecture in the area of network coding is true, then any constant
degree boolean
circuit for multiplication must have size $\Omega(n \lg n)$, thus
almost completely settling the complexity of multiplication
circuits. We additionally revisit classic conjectures in circuit complexity,
due to Valiant, and show that the network coding conjecture also
implies one of Valiant's conjectures.
\end{abstract}

\section{Introduction}
Multiplication is one of the most fundamental computational problems and the simple ``long multiplication'' $O(n^2)$-time algorithm for multiplying two $n$-digit numbers is taught to elementary school pupils around the world. Despite its centrality, the true complexity of multiplication remains elusive. In 1960, Kolmogorov conjectured that the thousands of years old $O(n^2)$-time algorithm is optimal and he arranged a seminar at Moscow State University with the goal of proving this conjecture. However only a week into the seminar, the student Karatsuba came up with an $O(n^{\lg_2 3}) \approx O(n^{1.585})$ time algorithm~\cite{Karatsuba:1962:MoMDNbAC}. The algorithm was presented at the next seminar meeting and the seminar was terminated. This sparked a sequence of improved algorithm such as the Toom-Cook algorithm~\cite{Toom:1963:TCoaSoFERtMoI,Cook:1966:OtmCToF} and the Sch\"{o}nhage-Strassen algorithm~\cite{Schonhage:1971:SMgZ}. The Sch\"{o}nhage-Strassen algorithm, as well as the current fastest algorithm by F\"{u}rer~\cite{Furer:2009:FIM}, are both based on the Fast Fourier Transform (FFT). F\"{u}rer's algorithm can be shown to run in time $O(n \lg n \cdot 4^{\lg^* n})$ when multiplying two $n$-bit numbers~\cite{Harvey:2018:FIMuSLV}. It can even be implemented as a constant degree Boolean circuit of the same size. Here $\lg^* n$ is the very slowly growing iterated logarithm. 

But what is the true complexity of multiplying two $n$-bit numbers? Can it be done via e.g. a Boolean circuit of size $O(n)$ like addition? Or is multiplication strictly harder? Our main contribution is to show a connection between multiplication and a central conjecture by Li and Li~\cite{lili} in the area of \emph{network coding}. Our results show that if the conjecture by Li and Li~\cite{lili} is true, then any constant degree Boolean circuit for computing the product of two $n$-bit numbers must have size $\Omega(n \lg n)$. This establishes a conditional lower bound for multiplication that comes within a $4^{\lg^*n}$ factor of F\"{u}rer's upper bound and implies that multiplication is strictly harder than addition.

Before diving into the details of our results, we first give a brief introduction to network coding.

\paragraph{Network Coding.}
Network coding studies communication problems in graphs. Given a graph $G$ with capacity constraints on the edges and $k$ data streams, each with a designated source-sink pair of nodes $(s_i,t_i)$ in $G$, what is the maximum rate at which data can be transmitted concurrently between the source-sink pairs? One solution is to just forward the data, which reduces the problem to a {\em multicommodity flow} problem. The central question in network coding is whether one can achieve a higher rate by using coding/bit tricks. This question is known to have a positive answer in directed graphs, where the rate increase may be as high as a factor $\Omega(|G|)$ (by sending XOR's of carefully chosen input bits), see e.g.~\cite{Adler:soda}. However the question remains wide open for undirected graphs where there are no known examples for which network coding can do better than the multicommodity flow rate. A central conjecture in network coding, due to Li an Li~\cite{lili}, says that coding yields no advantage in undirected graphs.
\begin{conjecture}[Undirected $k$-pairs Conjecture~\cite{lili}]
\label{con:undirected}
The coding rate is equal to the Multicommodity-Flow rate in undirected graphs.
\end{conjecture} 
Despite the centrality of this conjecture, it has heretofore resisted all attempts at either proving or refuting it. 
Conjecture~\ref{con:undirected} has been used twice before for proving lower bounds for computational problems. Adler \etal~\cite{Adler:soda} were the first to initiate this line of study. They presented conditional lower bounds for computing the transpose of a matrix via an \emph{oblivious algorithm}. Here oblivious means that the memory access pattern is fixed and independent of the input. Since a circuit is oblivious, they also obtain circuit lower bounds for matrix transpose. Very recently Farhadi \etal~\cite{FHLS18} showed how to remove the \emph{obliviousness} assumption for external memory problems. Their main result was a tight lower bound for external memory integer sorting, conditioned on Conjecture~\ref{con:undirected} being true.

\subsection{Our Results}

Our main result is an exciting new connection between network coding and the complexity of multiplication. Formally, we prove the following theorem:

\begin{theorem} \label{th:multiplicationLB}
Assuming Conjecture~\ref{con:undirected}, every boolean circuit with arbitrary gates and bounded in and out degrees that computes the product of two numbers given as two $n$-bit strings has size $\Omega(n \lg n)$.
\end{theorem}

In fact, we prove our $\Omega(n \lg n)$ lower bound for an even simpler problem than multiplication, namely the \emph{shift problem}: In the shift problem, we are given an $n$-bit string $x$ and an index $j \in [n]$. The goal is to construct a circuit that outputs the $2n$-bit string $y$ whose $i$th bit equals the $(i-j+1)$th bit of $x$ for every $j \le i \le j + n - 1$. Here we think of the index $j$ as being given in binary using $\lceil \lg n \rceil$ bits. We prove the following result:

\begin{theorem} \label{th:shiftLB}
Assuming Conjecture~\ref{con:undirected}, every boolean circuit with arbitrary gates and bounded in and out degrees that computes the shift problem has size $\Omega(n \lg n)$.
\end{theorem}

Theorem~\ref{th:multiplicationLB} follows as a corollary of Theorem~\ref{th:shiftLB} by observing that shifting $x$ by $j$ positions is equivalent to multiplication by $2^j$. Moreover, it is not hard to see that there is a linear sized circuit that has $\lceil \lg n \rceil$ input gates and $n$ output gates, where on an index $j \in [n]$, it outputs the number $2^j$ in binary (i.e. a single $1$-bit at position $j$).

We find it quite fascinating that even a simple instruction such as shifting requires circuits of size $\Omega(n \lg n)$, at least if we believe Conjecture~\ref{con:undirected}.

\paragraph{Valiant's Depth Reduction and Circuit Complexity Lower Bounds.}
In addition to our main lower bound results for multiplication, we also demonstrate that the network coding conjecture sheds new light on some fundamental conjectures by Valiant. 
In a 1977 survey %(that also included some unpublished results), 
Valiant~\cite{ValiantGraph} outlined potentially plausible attacks on the problem of 
proving a lower bound for the size of any circuit that can 
compute a permutation or even shifts of a given input. 
The goal was to prove that achieving both $O(n)$ size and $O(\lg n)$ depth 
for such circuits is impossible. 
While most of his attacks were rebuffed due to existence of 
complex and highly connected graphs that only had $O(n)$ edges (superconcentrators), 
Valiant outlined one last potential approach that could still be fruitful.
His main brilliant idea was to start with a circuit of some depth and by applying graph theoretical approaches
reducing the depth of the circuit while eliminating only a small number of edges. 
The hope was that information theoretical approaches could finish the job once the depth of the circuit was very low and once the (graph theoretical) complexity of the circuit was peeled away.

More formally, Valiant showed that for every circuit $C$ with $n$ input and output gates, of size $O(n)$, depth $O(\lg n)$ and fan-in $2$, and for every $\varepsilon > 0$, the function computed by $C$ can be computed by a boolean circuit with arbitrary gates $C'$ of depth $3$ with $n$ input and output gates and $\varepsilon n$ extra nodes. Moreover, the number of input gates directly connected to an output gate is bounded. That is, if we denote the set of input and output gates by $X$ and $Y$ respectively, then for every $y \in Y$, there are at most $O(n^{\varepsilon})$ wires connecting $y$ and $X$.

In turn, this reduction shows that it is enough to prove lower bounds on such depth $3$ circuits. 
Almost 20 years later and based on these ideas, Valiant~\cite{ValiantWhy} put forward several conjectures that if resolved could open the way 
for proving circuit complexity lower bounds. Loosely speaking, Valiant conjectured that if $\varepsilon \le 1/2$ then such depth $3$ circuits cannot compute cyclic-shift permutation. Before discussing Valiant's conjectures more formally, we first state our second main result, which essentially shows that Conjecture~\ref{con:undirected} implies one of Valiant's conjectures, albeit with a smaller (but still constant) bound on $\varepsilon$.

\begin{theorem} \label{th:depth3MultiplicationLB}
Let $C$ be a depth $3$ circuit that computes multiplication such that the following holds.
\begin{enumerate}
	\item The number of gates in the second layer of $C$ is at most $\varepsilon n$ for $\varepsilon \le 1/300$; and
	\item for every output gate $y$ of $C$, the number of input gates directly connected to $y$ is at most $c$.
\end{enumerate}
Then assuming Conjecture~\ref{con:undirected}, $c = \Omega\left(\frac{\lg n}{\lg \lg n}\right)$.
\end{theorem}
As with Theorem~\ref{th:multiplicationLB}, we prove Theorem~\ref{th:depth3MultiplicationLB} on an even restricted set of circuits, namely circuits that compute the shift function. We now turn to give a formal description of Valiant's Conjectures, and demonstrate how Theorem~\ref{th:depth3MultiplicationLB} brings us closer to settling them.

\paragraph{Valiant's Conjectures.} 
Let $\Gamma$ be a bipartite graph on two independent sets $X$ and $Y$ such that 
$X = \left\{x_1, \dots, x_n\right\}$ denotes a set of inputs and 
$Y = \left\{y_1, \ldots, y_n\right\}$ denotes a set of outputs. 
Furthermore assume, let $f_1, \ldots f_{\varepsilon n}$ be $\varepsilon n$ extra nodes and connect them by edges to all the nodes in $\Gamma$.
Denoting the resulting graph by $G$ consider all possible boolean circuits with arbitrary gates whose underlying topology is $G$. 
We say such a circuit computes a permutation $\pi\from Y \to X$ if for every assignment $x_1,\ldots,x_n \in \{0,1\}^n$ to the input gates, after the evaluation of the circuit $y_j$ is assigned $\pi(y_j)$ for every $j \in [n]$. 
Valiant conjectured that this should be impossible if $\varepsilon$ is too small or if $\Gamma$ has too few edges.
In particular, he proposed the following.
\begin{conjecture}\label{con:v1}
 If $\Gamma$ has maximum degree at most 3 and if $\varepsilon \le 1/2$, then there exists a permutation $\pi$ such that no circuit that has $G$ as its underlying topology can compute the permutation $\pi$. Moreover, there exists such $\pi$ that is a cyclic shift.
\end{conjecture}

Theorem~\ref{th:depth3MultiplicationLB} shows that conditioned on Conjecture~\ref{con:undirected}, if $\varepsilon \le 1/300$ then Valiant's first conjecture holds. We note that our proof for Theorem~\ref{th:depth3MultiplicationLB} continues to hold even if the gates' boolean functions are fixed after the shift offset is given. That is, if only the topology is fixed in advance. This coincides exactly with the formulation of Valiant's conjecture.
Valiant further conjectured the following.
\begin{conjecture}\label{con:v2}
 If $\Gamma$ has at most $n^{2-\delta}$ edges for some constant $\delta>0$, and if $\varepsilon \le 1/2$, 
 then there exists a permutation $\pi$ such that no circuit that has $G$ as its underlying topology can compute the permutation $\pi$. Moreover, there exists such $\pi$ that is a cyclic shift.
\end{conjecture}

\subsection{Related Work}

\paragraph{Lower Bounds for Multiplication.}
There are a number of previous lower bounds for multiplication in various restricted models of computation.
Clifford and Jalsenius~\cite{CJ11} considered a streaming variant of multiplication, where one number is fixed and the other is revealed one digit at a time. They require that a digit of the output is reported before the next digit of the input is revealed. In this streaming setting, they prove an $\Omega((\delta /w)n \lg n)$ lower bound, where $\delta$ is the number of bits in a digit and $w$ is the word size. For $\delta=1$ and $w=O(1)$, this is $\Omega(n \lg n)$. Ponzio~\cite{Ponz98} considered multiplication via read-once branching programs, i.e. programs that have bounded working memory and may only read each input bit exactly once. He proved that any read-once branching program for computing the middle bit of the product of two $n$-bit numbers, must use $\Omega(\sqrt{n})$ bits of working memory. Finally, we also mention the work of Morgenstern~\cite{Morgenstern:1973:NoaLBotLCotFFT} who proved lower bounds for computing the related FFT. Morgenstern proved an $\Omega(n \lg n)$ lower bound for computing the \emph{unnormalied} FFT via an arithmetic circuit when all constants used in the circuit are bounded. Unfortunately this doesn't say anything about the complexity of multiplying two $n$-bit numbers.

\paragraph{Valiant's Conjectures.}
Despite their importance, Valiant's conjectures are still mostly open. One interesting development by Riis~\cite{Riis2007}, shows that Conjecture~\ref{con:v2} as stated is incorrect. Riis proved that all cyclic shifts are realizable for $\varepsilon = \tfrac{1}{2} - \tfrac{1}{2n^{1-\delta}}$ where $n^{1 + \delta}$ is the total number of edges of $\Gamma$. Riis further conjectured that replacing the bound on $\varepsilon$ by a slightly stricter bound should result in a correct conjecture. Specifically, Riis suggest bounding $\varepsilon = \Theta\left(\frac{1}{\lg \lg n}\right)$.

\section{Preliminaries}
\label{sec:prelim}
We now give a formal definition of  Boolean circuits with arbitrary gates, followed by definitions of the $k$-pairs communication problem, the multicommodity flow problem. In the two latter problems we reuse some of the definitions used by Farhadi \etal~\cite{FHLS18}, which have been simplified a bit compared to the more general definition by Adler \etal~\cite{Adler:soda}. In particular, we have forced communication networks to be directed acyclic graphs. This is sufficient to prove our lower bounds and simplifies the definitions considerably.

\paragraph{Boolean Circuits with Arbitrary Gates.} 
A {\em Boolean Circuit with Arbitrary Gates} with $n$ source or input nodes and $m$ target or output nodes is a directed acyclic graph $C$ with $n$ nodes of in-degree $0$, which are called {\em input gates}, and are labeled with input variables $X=\{x_i\}_{i \in [n]}$ and $m$ nodes out-degree $0$, which are called {\em output gates} and are labeled with output variables $Y=\{y_i\}_{i \in [m]}$. All other nodes are simply called {\em gates}. For every gate $u$ of in-degree $k \ge 1$, $u$ is labeled with an arbitrary function $f_u:\{0,1\}^k \to \{0,1\}$. The circuit is also equipped with a topological ordering $v_1,\ldots,v_t$ of $C$, in which $v_i=x_i$ for $i \in [n]$ and $v_{t-i+1}=y_{m-i+1}$ for all $i \in [m]$. The {\em depth} of a circuit $C$ is the length of the longest path in $C$.
An {\em evaluation} of a circuit on an $n$ bit input $x=(x_1,\ldots,x_n) \in \{0,1\}^n$ is conducted as follows. For every $i \in [n]$, assign $x_j$ to $v_j$. For every $j \ge n+1$, assign to $v_j$ the value $f_{v_j}(u_1,\ldots,u_k)$, where $u_1,\ldots,u_k$ are the nodes of $C$ with edges going into $v_j$ in the order induced by the topological ordering. The {\em output} of $C$ on an $n$ bit input $x=(x_1,\ldots,x_n)$, denoted $C(x_1,\ldots,x_n)$ is the value assigned to $(y_1,\ldots,y_m)$ in the evaluation. We say a circuit computes a function $f : \{0,1\}^n \to \{0,1\}^m$ if for every $x=(x_1,\ldots,x_n) \in \{0,1\}^n$, $f(x_1,\ldots,x_n)=C(x_1,\ldots,x_n)$.

For every $j \in [t]$ and $b \in \{0,1\}$, we {\em hardwire} $b$ for $v_j$ in $C$ by removing $v_j$ and all adjacent edges from $C$, and replacing $v_j$ for $b$ in the evaluation of $f_{v_i}$ for every $i > j$ such that $v_jv_i$ is an edge in $C$.

\paragraph{$k$-Pairs Communication Problem.} 
The input to the $k$-pairs communication problem is a directed acyclic graph $G=(V,E)$ where each edge $e \in E$ has a capacity $c(e) \in \R^+$. There are $k$ sources $s_1,\dots,s_k \in V$ and $k$ sinks $t_1,\dots,t_k \in V$. 

Each source $s_i$ receives a message $A_i$ from a predefined set of messages $A(i)$. It will be convenient to think of this message as arriving on an in-edge. Hence we add an extra node $S_i$ for each source, which has a single out-edge to $s_i$. The edge has infinite capacity.

A network coding solution specifies for each edge $e \in E$ an alphabet $\Gamma(e)$ representing the set of possible messages that can be sent along the edge. For a node $v \in V$, define $\In(u)$ as the set of in-edges at $u$. A network coding solution also specifies, for each edge $e=(u,v) \in E$, a function $f_e : \prod_{e' \in \In(u)} \Gamma(e') \to \Gamma(e)$ which determines the message to be sent along the edge $e$ as a function of all incoming messages at node $u$. Finally, a network coding solution specifies for each sink $t_i$ a decoding function $\sigma_i : \prod_{e \in \In(t_i)} \Gamma(e) \to M(i)$. The network coding solution is correct if, for all inputs $A_1,\dots,A_k \in \prod_i A(i)$, it holds that $\sigma_i$ applied to the incoming messages at $t_i$ equals $A_i$, i.e. each source must receive the intended message.

In an execution of a network coding solution, each of the extra nodes $S_i$ starts by transmitting the message $A_i$ to $s_i$ along the edge $(S_i,s_i)$. Then, whenever a node $u$ has received a message $a_e$ along all incoming edges $e=(v,u)$, it evaluates $f_{e'}(\prod_{e \in \In(u)} a_e)$ on all out-edges and forwards the message along the edge $e'$.

We define the \emph{rate} of a network coding solution as follows: Let each source receive a uniform random and independently chosen message $A_i$ from $A(i)$. For each edge $e$, let $A_e$ denote the random variable giving the message sent on the edge $e$ when executing the network coding solution with the given inputs. The network coding solution achieves rate $r$ if:
\begin{itemize}
\item $H(A_i) \geq r$ for all $i$.
\item For each edge $e \in E$, we have $H(A_e) \leq c(e)$.
\end{itemize}
Here $H(\cdot)$ denotes binary Shannon entropy. The intuition is that the rate is $r$, if the solution can handle sending a message of entropy $r$ bits between every source-sink pair.

\paragraph{Multicommodity Flow.}
A multicommodity flow problem in an undirected graph $G=(V,E)$ is specified by a set of $k$ source-sink pairs $(s_i,t_i)$ of nodes in $G$. We say that $s_i$ is the source of commodity $i$ and $t_i$ is the sink of commodity $i$. Each edge $e \in E$ has an associated capacity $c(e) \in \R^+$. 
\
A (fractional) solution to the multicommodity flow problem specifies for each pair of nodes $(u,v)$ and commodity $i$, a flow $f^i(u,v) \in [0,1]$. Intuitively $f^i(u,v)$ specifies how much of commodity $i$ that is to be sent from $u$ to $v$. The flow satisfies \emph{flow conservation}, meaning that:
\begin{itemize}
\item For all nodes $u$ that is not a source or sink, we have $\sum_{w \in V} f^i(u,w) - \sum_{w \in V} f^i(w,u) = 0$.
\item For all sources $s_i$, we have $\sum_{w \in V} f^i(s_i,w) - \sum_{w \in V}f^i(w,s_i) = 1$.
\item For all sinks we have $\sum_{w \in V} f^i(w,t_i) - \sum_{w \in V} f^i(t_i,w) = 1$.
\end{itemize}
The flow also satisfies that for any pair of nodes $(u,v)$ and commodity $i$, there is only flow in one direction, i.e. either $f^i(u,v)=0$ or $f^i(v,u)=0$. Furthermore, if $(u,v)$ is not an edge in $E$, then $f^i(u,v) = f^i(v,u)=0$. A solution to the multicommodity flow problem achieves a rate of $r$ if:
\begin{itemize}
\item For all edges $e=(u,v) \in E$, we have $r \cdot \sum_i  (f^i(u,v) + f^i(v,u)) \leq c(e)$.
\end{itemize}
Intuitively, the rate is $r$ if we can handle a demand of $r$ for every commodity.

\paragraph{The Undirected $k$-Pairs Conjecture.} Conjecture~\ref{con:undirected} implies the following for our setting: Given an input to the $k$-pairs communication problem, specified by a directed acyclic graph $G$ with edge capacities and a set of $k$ source-sink pairs, let $r$ be the best achievable network coding rate for $G$. Similarly, let $G'$ denote the undirected graph resulting from making each directed edge in $G$ undirected (and keeping the capacities and source-sink pairs). Let $r'$ be the best achievable flow rate in $G'$. Conjecture~\ref{con:undirected} implies that $r \leq r'$.

Having defined coding rate and flow rate formally, we also mention that a result of Braverman \etal~\cite{braverman2016network} implies that if there exists a graph $G$ where the network coding rate $r$, and the flow rate $r'$ in the corresponding undirected graph $G'$, satisfies $r \geq (1+\eps)r'$ for a constant $\eps>0$, then there exists an infinite family of graphs $\{G^*\}$ for which the corresponding gap is at least $(\lg |G^*|)^c$ for a constant $c>0$. So far, all evidence suggest that no such gap exists, as formalized in Conjecture~\ref{con:undirected}.

\section{Key Tools and Techniques}
The main idea in the heart of both proofs is the simple fact that in a graph with $t$ vertices and maximum degree at most $c$, most node pairs lie far away from one another. Specifically, for every node $u$ in $G$, at least $t - \sqrt{t}$ nodes have distance $\ge \tfrac{1}{2}\log_c t$ from $u$. While this key observation is almost enough to prove Theorem~\ref{th:shiftLB}, the proof of Theorem~\ref{th:depth3MultiplicationLB} requires a much more subtle approach, as there is no bound on the maximum degree in the circuits in question. The only bound we have is on the number of wires going directly between from input gates into output gates. Specifically, every two nodes in the underlying undirected graph are at distance $\le 3$ (see figure~\ref{fig:circuit}).

In order to overcome this obstacle, we present a construction of a communication network based on the circuit $C$ that essentially eliminates the middle layer in the depth-$3$ circuit $C$, thus leaving a bipartite graph with bounded maximum degree.
To this end, we observe that since the size of the middle layer is bounded by $\varepsilon n$, then there exists a large set ${\cal F}$ of inputs in $\{0,1\}^n$ such that on all inputs from ${\cal F}$, the gates $f_1,\ldots,f_{\varepsilon n}$ attain the same values. By hardwiring these values to the circuit, we can evaluate the circuit for all inputs in ${\cal F}$ on a depth-$2$ circuit $\Gamma$ obtained from $C$ by removing $f_1,\ldots, f_{\varepsilon n}$. We next turn to construct the communication network.
Employing ideas recently presented by Farhadi \etal~\cite{FHLS18}, we "wrap" the depth-$2$ circuit by adding source and target nodes. In order to cope with inputs that do not belong to ${\cal F}$, we add a designated {\em supervisor} node $u$ (see figure~\ref{fig:network}). Loosely speaking, the source nodes transmit their input to $u$, and $u$ sends back the information needed to "edit" the input string $x$ and construct an input string $x' \in {\cal F}$, which is then transferred to the circuit $\Gamma$ as blackbox.

\paragraph{The Correction Game.} In order to bound the edge capacities of the network $G$ in a way that the supervisor node can transmit enough information to achieve a high communication rate, but then again not allow to much flow to go through the supervisor when considering $G$ as a multicommodity flow instance, Farhadi \etal~\cite{FHLS18} defined a game between a set of $m$ players and a supervisor, where given a fixed set ${\cal F} \subseteq \{0,1\}^n$ and a random string $\beta \in \{0,1\}^n$ given as a concatenation of $m$ strings $\beta_1,\ldots,\beta_m$ of length $n/m$ each, the goal is to "correct" $x$ and produce a string $\chi \in \{0,1\}^n$ such that $\beta \oplus \chi \in {\cal F}$. The caveat is that the only communication allowed is between the players and the supervisor. That is, no communication, and thus no cooperation, is allowed between the $m$ players. Formally, the game is defined as follows.

\begin{definition}
Let ${\cal F} \subseteq \{0,1\}^n$. The {\em ${\cal F}$-correction game} with $m+1$ players is defined as follows. The game is played by $m$ ordinary players $p_1,\ldots, p_m$ and one designated {\em supervisor} player $u$.  The supervisor $u$ receives $m$ strings $\beta_1,\ldots,\beta_m \in \{0,1\}^{n/m}$ chosen independently at random. For every $\ell \in [m]$, $u$ then sends $p_\ell$ a message $R_\ell$. Given $R_\ell$, the player $p_\ell$ produces a string $\chi_\ell \in \{0,1\}^{n/m}$ such that $(\beta_1 \oplus \chi_1)\circ(\beta_2 \oplus \chi_2)\circ(\beta_m \oplus \chi_m) \in {\cal F}$.
\end{definition}

Farhadi \etal additionally present a protocol for the ${\cal F}$-correction game in which the supervisor player sends prefix-free messages to the $m$ players, and moreover, they give a bound on the amount of communication needed as a function of the number of players and the size of ${\cal F}$. 
\begin{lemma} [\cite{FHLS18}]\label{l:protocolBound}
If $|{\cal F}| \ge 2^{(1-\varepsilon)n}$, then there exists a protocol for the ${\cal F}$-correction game with $m+1$ players such that the messages $\{R_\ell\}_{\ell \in [m]}$ are prefix-free and
$$\sum_{\ell \in [m]}{\mathbb{E}[|R_\ell|]} \le 3m + 2m\lg\left(\sqrt{\frac{\varepsilon}{2}}\cdot \frac{n}{m} + 1\right) + \sqrt{\frac{\varepsilon}{8}}\cdot n \lg\frac{2}{\varepsilon} \;,$$
\end{lemma}

\section{A Lower Bound for Boolean Circuits Computing Multiplication}
In this section we show that conditioned on Conjecture~\ref{con:undirected}, every bounded degree circuit computing multiplication must have size at least $\Omega(n \lg n)$, thus proving Theorems~\ref{th:multiplicationLB} and~\ref{th:shiftLB}.
In fact, we will prove something slightly stronger. Define the shift function $s : \{0,1\}^n\times[n] \to \{0,1\}^{2n}$ as follows. For every $x=(x_1,\ldots,x_n) \in \{0,1\}^n$ and $\ell \in [n]$, $s(x,\ell) = (y_1,\ldots,y_{2n})$ where $y_j = x_{j-\ell+1}$ if $\ell \le j \le \ell+n-1$ and $y_j=0$ otherwise. We will show that every circuit with bounded in and out degrees that computes the shift function on $n$-bit numbers has size $\Omega(n \lg n)$. Clearly, a circuit that can compute the product of two $n$-bit numbers can also compute the shift function. Let $c$ denote the maximum in and out degree in $C$, and let $j \in [n]$. Then in the undirected graph induced by $C$, there are at most $\sqrt{n}$ nodes whose distance from $x_j$ is at most $\frac{1}{2}\log_{2c}n$. Therefore among $y_j,\ldots,y_{j+n-1}$, at least $n - \sqrt{n} - 1 \ge n - 2\sqrt{n}$ are at distance at least $\frac{1}{2}\log_{2c}n$. In other words, $\Pr_{\ell \in [n]}[d_{\hat{C}}(x_j, y_{j+\ell-1}) \ge \frac{1}{2}\log_{2c}n] \ge 1-\frac{2}{\sqrt{n}}$, where $\hat{C}$ denotes the undirected graph induced by $C$ (by removing edge directions). Therefore there exists a shift $\ell_0 \in [n]$ such that $|\{j \in [n] : d_{\bar{C}}(x_j, y_{j+\ell_0-1}) \ge \frac{1}{2}\log_{2c}n\}| \ge n-2\sqrt{n} \ge n/2$. 

Fixing $\ell_0$, let consider the following communication problem. For each $j \in [n]$, $s_j=x_j \in_R \{0,1\}$ and $t_j = y_{j+\ell_0-1}$. The circuit $C$ equipped with $1$-uniform edge capacities is a network coding solution to this problem with rate $r \ge 1$. By the undirected $n$-pairs conjecture, there is a multicommodity flow in $\hat{C}$ that transfers one unit of flow from each source to its corresponding sink. For every $j$, let $f^j : E \to [0,1]$ be the flow associated with commodity $j$. Then
$$|E| = \sum_{e \in E}{c_e} \ge \sum_{e \in E}{\sum_{j \in [n]}{f^j(e)}} \ge \Omega(n \log_c n) \;.$$

\section{A Lower Bound for Depth \texorpdfstring{$3$}{3} Boolean Circuits Computing Multiplication}

\begin{figure}[t]
	\centering
		\includegraphics[scale=0.7]{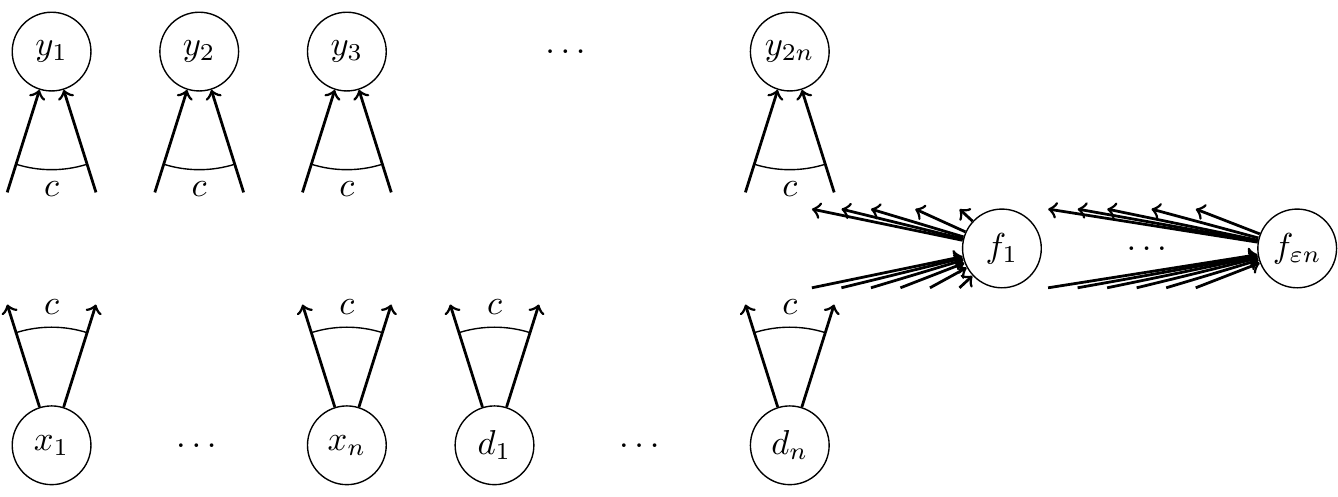}
   \caption{The depth $3$ circuit $C$.}
   \label{fig:circuit}
\hrule
\end{figure}

Let $C$ be a depth $3$ circuit that computes multiplication such that the number of gates in the second layer of $C$ is at most $\varepsilon n$ for some small $\varepsilon \in (0,1)$ and for every $u \in Y$, $deg_{\bar{C}[X \cup Y]}(u) \le c$, where once again $\bar{C}$ denotes the undirected graph induced by $C$, and $\bar{C}[X \cup Y]$ is the subgraph of $\bar{C}$ induced by $X \cup Y$. By slightly increasing $c$ and $\varepsilon$ (by a small constant factor) and without loss of generality, we can assume that this applies for all $u \in X$ as well.

Denote the input and output gates of $C$ by $X = \{x_1,\ldots,x_n,\hat{x}_1,\ldots,\hat{x}_n\}$ and $Y = \{y_1,\ldots,y_{2n}\}$ respectively, and denote the set of the middle-layer gates by $F = \{f_1,\ldots,f_{\varepsilon n}\}$ (see Figure~\ref{fig:circuit}).

\begin{figure}[t]
	\centering
		\includegraphics[scale=0.7]{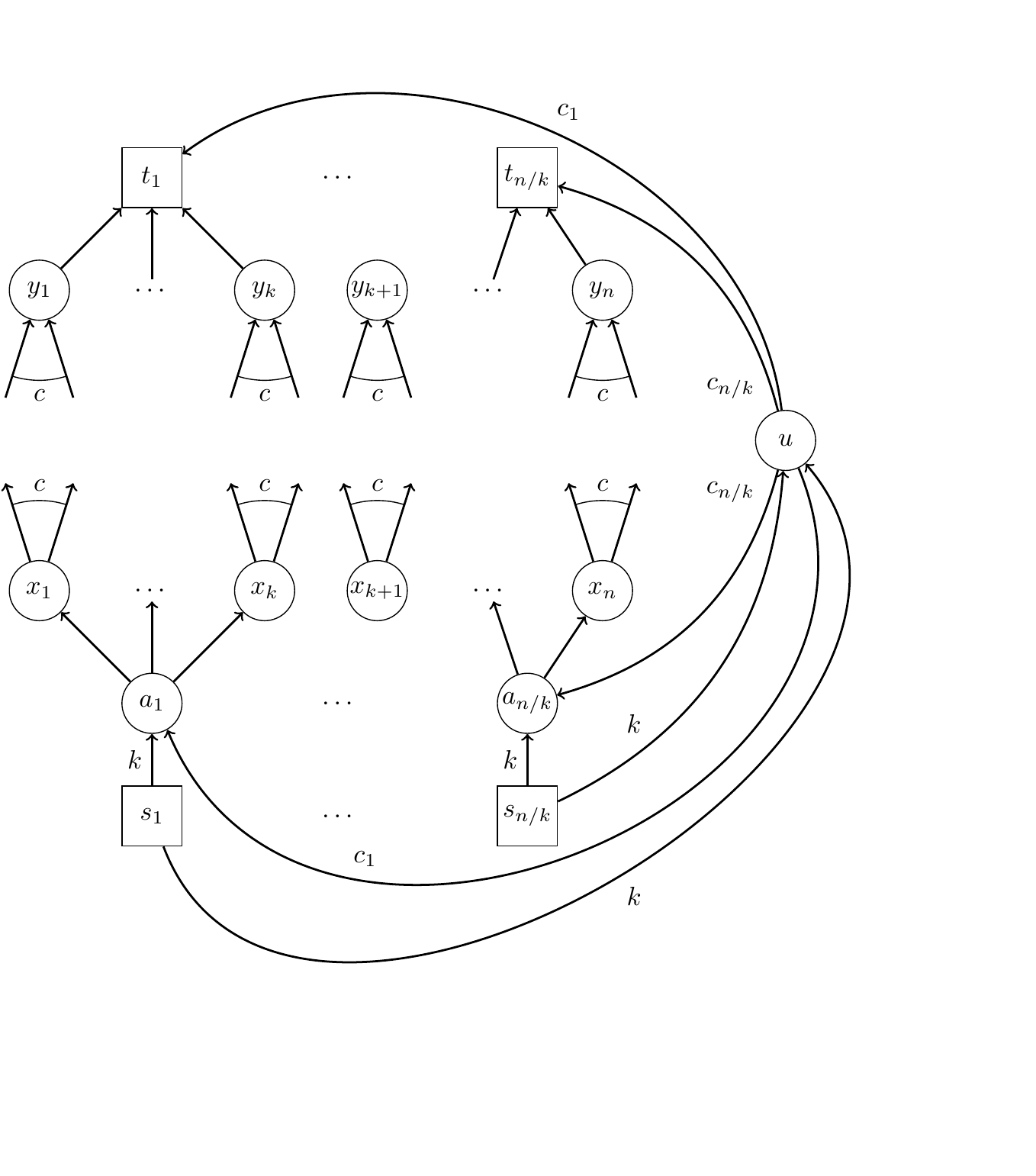}
   \caption{Given the $2$-layer circuit $\Gamma$ spanned by $x_1,\ldots,x_n,y_1,\ldots,y_n$, we construct the communication network graph $G$.}
   \label{fig:network}
\hrule
\end{figure}
As before, we focus on computing the shift function, thus limiting the input to $(\hat{x}_1,\ldots,\hat{x}_n)$ to have exactly one $1$-entry.
We next partition $(x_1,\ldots,x_n)$ into consecutive blocks of size $k=20$ bits each.  For every $\ell \in [n/k]$ let $B_\ell = \{k(\ell-1)+1,\ldots,k\ell\}$ be the set of indices belonging to the $\ell$th block. 
\begin{definition}
For every $\alpha \in [n]$ and $\ell \in [n/k]$, we say $B_\ell$ is {\em far from all targets} (with respect to $\alpha$) if for all sources in the block are at distance at least $\frac{1}{2}\log_{2c}n$ from all respective destinations in $\bar{C}[X \cup Y]$. That is for every $u,v \in B_\ell$, $d_{\bar{C}[X \cup Y]}(x_u,y_{v+\alpha-1}) \ge \frac{1}{2}\log_{2c}n$.
\end{definition}
Let $\alpha \in_R [n]$. By the constraint on the degrees, for every $j \in [n]$, there are at most $\sqrt{n}$ nodes whose distance from $x_j$ is at most $\frac{1}{2}\log_{2c}n$ in $\bar{C}[X \cup Y]$. Therefore for every $\ell \in [n/k]$, 
$$\Pr_{\alpha \in_R[n]}\left[B_\ell \; \text{is far from all targets}\right] \ge 1 - \frac{k^2}{\sqrt{n}} \;.$$
By averaging we get that for large enough $n$ there is some $\alpha_0 \in [n]$ such that there are at least $\frac{n}{k} - k\sqrt{n} \ge \frac{9n}{10k}$ blocks which are far from all targets. Without loss of generality, we may assume for ease of notation that $\alpha_0=1$. 
By hardwiring $1$ for $\alpha_0$ into the circuit $C$, the circuit now simply transfers $(x_1,\ldots,x_n)$ to $(y_1,\ldots,y_n)$. 

\paragraph{Reduction to Network Coding.}  Let $x = (x_1,\ldots,x_n)$ and $i \in [\varepsilon n]$. By slightly abusing notation, we denote the value of the gate $f_i$ when evaluating the circuit by $f_i(x_1,\ldots,x_n)$.
By averaging, there exist a string $(\hat{f}_1, \ldots, \hat{f}_{\varepsilon n})$ and a set ${\cal F} \subseteq \{0,1\}^n$ such that $|{\cal F}| \ge 2^{(1-\varepsilon) n}$ and such that for every $x = (x_1,\ldots,x_n) \in {\cal F}$ and $i \in [\varepsilon n]$, $f_i(x_1,\ldots,x_n) = \hat{f}_i$.
By hardwiring $(\hat{f}_1, \ldots, \hat{f}_{\varepsilon n})$ for $(f_1,\ldots,f_n)$ into the circuit $C$, we get a new circuit denoted $\Gamma$ that contains only the input and output gates of $C$, and transfers $(x_1,\ldots,x_n)$ to $(y_1,\ldots,y_n)$ for every $(x_1,\ldots,x_n) \in {\cal F}$. Moreover, the set of edges between $X$ and $Y$ in $\Gamma$ is equal to the set of edges between $X$ and $Y$ in $C$.

Next, we construct a communication network $G$ by adding some nodes and edges to $\Gamma$, as demonstrated also in Figure~\ref{fig:network}. We add a new set of nodes $\{s_j,a_j,t_j\}_{j=1}^{n/k} \cup \{u\}$. For every $\ell \in [n/k]$, add edges $s_{\ell}a_{\ell}$ and $s_{\ell}u$ of capacity $k$ and edges $ua_{\ell}$ and $ut_{\ell}$ of capacity $c_{\ell} = \mathbb{E}[|R_\ell|]$, where $R_\ell$ is the message sent to player $p_\ell$ by the supervisor player in the ${\cal F}$-correction game protocol for $n/k+1$ players guaranteed in Lemma~\ref{l:protocolBound}. In addition, for every $\ell \in [n/k]$ and every $j \in B_{\ell}$ add edges $a_{\ell}x_j$ and $y_jt_{\ell}$ of capacity $1$. All edges of $\Gamma$ are assigned capacity of $1$.

\paragraph{Transmitting Data.} In what follows, we will lower bound the communication rate of the newly constructed network $G$.
\begin{lemma}\label{l:comRate}
There exists a network coding solution on $G$ that achieves rate $k$.
\end{lemma}
To this end, let $A_1,\ldots,A_{n/k} \in \{0,1\}^k$ be independent uniform random variables. We next give a protocol by which the sources $s_1,\ldots,s_{n/k}$ transmit $A_1,\ldots,A_{n/k}$ to the targets $t_1,\ldots,t_{n/k}$. The protocol employs as a an intermediate step the correction game protocol guaranteed by Lemma~\ref{l:protocolBound}.
\begin{enumerate}
	\item For every $\ell \in [n/k]$, $s_\ell$ sends $A_\ell$ to $a_\ell$ over the edge $s_\ell a_\ell$ and to $u$ over the edge $s_\ell u$.
	\item Employing the ${\cal F}$-correction game protocol with $n/k+1$ players, for every $\ell \in [n/k]$, $u$ sends a message $R_\ell$ to $a_\ell$ over the edge $ua_\ell$ and to $t_\ell$ over the edge $ut_\ell$. Following the correction game protocol, for every $\ell$, given $R_\ell$, $a_\ell$ and $t_\ell$ produce a string $\chi_\ell$ satisfying that $(A_1\oplus \chi_1)\circ\ldots\circ(A_{n/k}\oplus \chi_{n/k}) \in {\cal F}$.
	\item For every $\ell \in [n/k]$ and every $i \in [k]$, $a_\ell$ transmits the $i$th bit of $A_\ell \oplus \chi_\ell$ to the $i$th gate in the $\ell$th block, namely $x_{(\ell-1)k+i}$. Note that $(x_1,\ldots,x_n) = (A_1\oplus \chi_1)\circ\ldots\circ(A_{n/k}\oplus \chi_{n/k}) \in {\cal F}$.
	\item Next, the communication network employs the circuit $\Gamma$ and transmits $(x_1,\ldots,x_n)$ to $(y_1,\ldots,y_n)$. For every $\ell \in [n/k]$ and every $i \in B_\ell$, $y_i$ transmits $x_i$ to $t_\ell$. 
	\item Finally, for every $\ell \in [n/k]$, $t_\ell$ now holds both $A_\ell \oplus \chi_\ell$ and $\chi_\ell$. Therefore $t_\ell$ can recover $A_\ell$.
\end{enumerate}
By invoking the protocol described above, every one of the $n/k$ sources sends $k$ bits to the corresponding target. For every edge $e \in G$, let $A_e$ denote the random variable giving the message sent on the edge $e$ when executing the protocol. 
\begin{claim}
For every $e \in G$, $H(A_e) \le c_e$.
\end{claim}
\begin{proof}
First note that for every $\ell \in [n/k]$, every edge $e$ leaving $s_\ell$ has capacity $k$ and transmits $A_\ell$. Therefore $H(A_\ell) = k \le c_e$.	
Every edge $e$ that is not leaving any source nor $u$ has capacity $1$ and transmits exactly one bit (not necessarily uniformly random) of information. Therefore $c_e = 1 \ge H(A_e)$.
Finally, let $e$ be an edge leaving $u$. Then there exists some $\ell \in [n/k]$ such that $e=ua_\ell$ or $e=ut_\ell$. In both cases the message transmitted on $e$ is $R_\ell$ and the capacity $c_e$ of $e$ satisfies $c_e = c_\ell = \mathbb{E}[|R_\ell|] \ge H(R_\ell)$, where the last inequality follows from Shannon's Source Coding theorem, as all messages are prefix-free.
\end{proof}
We can therefore conclude that the network $G$ achieves rate $\ge k$, and the proof of Lemma~\ref{l:comRate} is complete.
\paragraph{Deriving the Lower Bound.} By Conjecture~\ref{con:undirected}, the underlying undirected graph $\bar{G}$ achieves a multicommodity-flow rate $\ge k$. Therefore there exists a multicommodity flow $\{f^{\ell}\}_{\ell \in [n/k]} \subseteq [0,1]^{E(\bar{G})}$ that achieves rate $k$. We first observe that at most a constant fraction of the flow can go through the supervisor node $u$. To see this, we note that as $|{\cal F}| \ge 2^{(1-\varepsilon)n}$, then by Lemma~\ref{l:protocolBound} the expected total information sent by the supervisor in the ${\cal F}$-correction game with $n/k$ players is at most 
\begin{equation}
\frac{3n}{k} + \frac{2n}{k}\lg\left(k\sqrt{\frac{\varepsilon}{2}} + 1\right) + \sqrt{\frac{\varepsilon}{8}} \cdot n \lg\frac{2}{\varepsilon} \le \frac{5n}{k}
\label{eq:protocol}
\end{equation}
Therefore by the definition of the capacities $\{c_\ell\}_{\ell \in [n/k]}$ we get that for small enough (constant) $\varepsilon$,
\begin{equation}
\sum_{\ell \in [n/k]}{c_{ua_\ell}} = \sum_{\ell \in [n/k]}{c_{ut_\ell}} = \sum_{\ell \in [n/k]}{c_{\ell}} \le \frac{5n}{k}
\label{eq:capacities}
\end{equation}
Since $\{f^{\ell}\}_{\ell \in [n/k]}$ achieves rate $k$ we conclude that 
\begin{equation*}
\begin{split}
 k \cdot\sum_{v \in V(\bar{G}) : uv \in E(\bar{G})}{ \sum_{\ell \in [n/k]}{(f^{\ell}(u,v) + f^{\ell}(v,u))}} &\le \sum_{v \in V(\bar{G}) : uv \in E(\bar{G})}{c_e}\\
&= \sum_{\ell \in [n/k]}{c_{us_\ell}} + \sum_{\ell \in [n/k]}{(c_{ua_\ell}+c_{ut_\ell})}\le n + \frac{10n}{k}\;,
\end{split}
\label{eq:uCapacity}
\end{equation*}
and therefore 
\begin{equation}
\sum_{v \in V(\bar{G}) : uv \in E(\bar{G})}{ \sum_{\ell \in [n/k]}{(f^{\ell}(u,v) + f^{\ell}(v,u))}} \le \frac{n}{k} + \frac{10n}{k^2} \le 1.5\frac{n}{k} \;.
\label{eq:uTotal}
\end{equation}
By the flow-conservation constraint, we know that 
therefore the total amount of flow that can go through $u$ is $\le 0.75 \frac{n}{k}$. By averaging, at least a $1/6$ fraction of the sources send at least $1/10$ units of flow through $\bar{G} - u$. By the choice of $\alpha_0$, in $\bar{G}-u$, at least a $1/15$ of the sources are at least $\frac{1}{2}\log_{2c}(n)$ away from their targets. Without loss of generality, assume these are the first $\tfrac{n}{15k}$ sources. We conclude that 
\begin{equation}
\begin{split}
cn \ge |E[X \cup Y]| &= \sum_{e \in E[X \cup Y]}{c_e} \ge k \cdot \sum_{e = vw \in E[X \cup Y]}{\sum_{\ell \in [n/k]}{f^{\ell}(v,w)+f^{\ell}(w,v)}} \\
&\ge k \cdot \sum_{\ell \in [n/15k]}{\sum_{e = vw \in E[X \cup Y]}{f^{\ell}(v,w)+f^{\ell}(w,v)}} \ge \frac{n}{30}\log_{2c}(n) \;,
\end{split}
\end{equation}
and therefore $c \ge \Omega\left(\frac{\lg n}{\lg \lg n}\right)$, and the proof of Theorem~\ref{th:depth3MultiplicationLB} is now complete.

\subsection{Remarks and Extensions}
For sake of fluency, some minor remarks and extensions were intentionally left out of the text, and will be discussed now. 
\paragraph{Circuits with Bounded Average Degree.} Our results still hold if we relax the second requirement of Theorem~\ref{th:depth3MultiplicationLB} and require instead that the number of edges in $\bar{C}[X \cup Y]$ is at most $cn$. That is, the average degree in $\bar{C}[X \cup Y]$ is at most $c$. To see this, note that under this assumption, there are at most $0.001n$ gates in $X \cup Y$ whose degree in $\bar{C}[X \cup Y]$ is larger than $1000c$. For each such gate $v$, add a new node $f$ in the middle layer, and connect $v$ and all the neighbours of $v$ in $\bar{C}[X \cup Y]$ to $f$. Then delete all the edges adjacent to $v$ in $\bar{C}[X \cup Y]$. The number of nodes added to the middle layer is at most $0.001n$, and the degree of all nodes in $\bar{C}[X \cup Y]$ is now bounded by $1000c$. The rest of our proof continues as before.

\paragraph{Shifts vs. Cyclic Shifts.} In order to prove lower bounds for circuits computing multiplication, our results are stated in terms of shifts (which are a special case of products, as mentioned). This is in contrast to Valiant's conjectures, which are stated in terms of cyclic shifts. However, we draw the readers attention to the fact that our proofs work for cyclic shifts as well. The exact same arguments apply, and the proofs remain unchanged.

\newcommand{\etalchar}[1]{$^{#1}$}


\begin{thebibliography}{AHJ{\etalchar{+}}06}

\bibitem[AHJ{\etalchar{+}}06]{Adler:soda}
M.~Adler, N.~J.~A. Harvey, K.~Jain, R.~Kleinberg, and A.~R. Lehman.
\newblock On the capacity of information networks.
\newblock In {\em Proceedings of the Seventeenth Annual ACM-SIAM Symposium on
  Discrete Algorithm}, SODA '06, pages 241--250. Society for Industrial and
  Applied Mathematics, 2006.
\newblock Available from:
  \url{http://dl.acm.org/citation.cfm?id=1109557.1109585}.

\bibitem[BGS17]{braverman2016network}
M.~Braverman, S.~Garg, and A.~Schvartzman.
\newblock Coding in undirected graphs is either very helpful or not helpful at
  all.
\newblock In {\em 8th Innovations in Theoretical Computer Science Conference,
  {ITCS} 2017, January 9-11, 2017, Berkeley, CA, {USA}}, pages 18:1--18:18,
  2017.

\bibitem[CJ11]{CJ11}
R.~Clifford and M.~Jalsenius.
\newblock Lower bounds for online integer multiplication and convolution in the
  cell-probe model.
\newblock In {\em Automata, Languages and Programming - 38th International
  Colloquium, {ICALP} 2011, Zurich, Switzerland, July 4-8, 2011, Proceedings,
  Part {I}}, pages 593--604, 2011.

\bibitem[Coo66]{Cook:1966:OtmCToF}
S.~A. Cook.
\newblock {\em On the minimum computation time of functions}.
\newblock PhD thesis, Harvard University, 1966.

\bibitem[Fü09]{Furer:2009:FIM}
M.~Fürer.
\newblock Faster integer multiplication.
\newblock {\em SIAM Journal on Computing}, 39(3):979--1005, 2009.
\newblock \href {http://dx.doi.org/10.1137/070711761}
  {\path{doi:10.1137/070711761}}.

\bibitem[FHLS19]{FHLS18}
A.~Farhadi, M.~Hajiaghayi, K.~G. Larsen, and E.~Shi.
\newblock Lower bounds for external memory integer sorting via network coding.
\newblock In {\em Proceedings of the 52st Symposium on Theory of Computing,
  {STOC} 2019}, 2019.
\newblock To appear.

\bibitem[HvdH18]{Harvey:2018:FIMuSLV}
D.~Harvey and J.~van~der Hoeven.
\newblock Faster integer multiplication using short lattice vectors.
\newblock {\em CoRR}, 2018.
\newblock \href {http://arxiv.org/abs/1802.07932} {\path{arXiv:1802.07932}}.

\bibitem[KO62]{Karatsuba:1962:MoMDNbAC}
A.~A. Karatsuba and Y.~P. Ofman.
\newblock Multiplication of many-digital numbers by automatic computers.
\newblock {\em Proceedings of the {USSR} Academy of Sciences}, 145:293--294,
  1962.

\bibitem[LL04]{lili}
Z.~Li and B.~Li.
\newblock Network coding: The case of multiple unicast sessions.
\newblock In {\em Proceedings of the 42nd Annual Allerton Conference on
  Communication, Control, and Computing}, 2004.

\bibitem[Mor73]{Morgenstern:1973:NoaLBotLCotFFT}
J.~Morgenstern.
\newblock Note on a lower bound on the linear complexity of the fast {F}ourier
  transform.
\newblock {\em Journal of the ACM}, 20(2):305--306, 1973.
\newblock \href {http://dx.doi.org/10.1145/321752.321761}
  {\path{doi:10.1145/321752.321761}}.

\bibitem[Pon98]{Ponz98}
S.~Ponzio.
\newblock A lower bound for integer multiplication with read-once branching
  programs.
\newblock {\em {SIAM} J. Comput.}, 28(3):798--815, 1998.

\bibitem[Rii07]{Riis2007}
S.~Riis.
\newblock Information flows, graphs and their guessing numbers.
\newblock {\em The Electronic Journal of Combinatorics}, 14(1), 2007.

\bibitem[SS71]{Schonhage:1971:SMgZ}
A.~Sch{\"o}nhage and V.~Strassen.
\newblock Schnelle multiplikation gro{\ss}er zahlen.
\newblock {\em Computing}, 7(3):281--292, Sep 1971.
\newblock \href {http://dx.doi.org/10.1007/BF02242355}
  {\path{doi:10.1007/BF02242355}}.

\bibitem[Too63]{Toom:1963:TCoaSoFERtMoI}
A.~L. Toom.
\newblock The complexity of a scheme of functional elements realizing the
  multiplication of integers.
\newblock {\em Proceedings of the {USSR} Academy of Sciences}, 150(3):496--498,
  1963.

\bibitem[Val77]{ValiantGraph}
L.~G. Valiant.
\newblock Graph-theoretic arguments in low-level complexity.
\newblock In {\em Mathematical Foundations of Computer Science 1977}, pages
  162--176, 1977.

\bibitem[Val92]{ValiantWhy}
L.~G. Valiant.
\newblock Why is boolean complexity theory difficult?
\newblock In {\em Proceedings of the London Mathematical Society Symposium on
  Boolean Function Complexity}, pages 84--94, 1992.

\end{thebibliography}
\end{document}